\documentclass[12pt]{amsart}
\RequirePackage{amsfonts,amsrefs,a4wide} \RequirePackage{amssymb}
\RequirePackage{eucal,tikz}  
\usepackage{ amscd, amsfonts,mathrsfs,url}

\newtheorem{thm}{Theorem}

\newtheorem{defn}[thm]{Definition}
\theoremstyle{remark}

\newcommand{\A}{\mathscr A}

\renewcommand{\epsilon}{\varepsilon}

\newcommand{\X}{\mathscr X}

\newcommand{\Zpr}{\mathscr Z}

\newcommand{\Z}{\mathbb Z}
\newcommand{\Zd}{{\mathbb Z^d}}

\newcommand{\E}{\mathbb E}

\renewcommand{\H}{\text{\fontencoding{OT1}\fontfamily{phv}\fontseries{em}\selectfont{}H}}
\newcommand{\h}
{\text{\fontencoding{OT1}\fontfamily{phv}\fontseries{em}\selectfont{}h}}
\renewcommand{\P}{\mathbb P}

\def\mybf #1{{\bf{\emph{#1}}}}

\begin{document}
\title{Erasure entropies and Gibbs Measures}
\author{Aernout van Enter}
\address{Aernout van Enter: Johann Bernoulli Institute for
   Mathematics and Computer Science, University of Groningen, PO
Box 407, 9700 AK, Groningen, The Netherlands}
\email{aenter@phys.rug.nl}
\author{Evgeny Verbitskiy}
\address{Evgeny Verbitskiy: Mathematical Institute, University of Leiden, PO Box 9512, 2300 RA Leiden, The Netherlands}
\email{evgeny@math.leidenuniv.nl}
\address{Johann Bernoulli Institute for
   Mathematics and Computer Science, University of Groningen, PO
Box 407, 9700 AK, Groningen, The Netherlands}
\email{e.a.verbitskiy@rug.nl}
 
\begin{abstract}
Recently Verd\`u and Weissman introduced {\em erasure entropies},
which are meant to measure the information carried by
one or more symbols given all of the remaining symbols in the realization
of a random process or field.
A natural relation to {\em Gibbs measures} has also been
observed. In this short note we study this relation further, review a few 
earlier contributions from statistical mechanics, 
and provide the formula for the erasure entropy of a Gibbs measure
 in terms of the corresponding  potential.
 For some  $2$-dimensional Ising models, for which Verd\'u and Weissman 
suggested a numerical proedure, we show how to obtain  an exact formula
 for the erasure entropy.

\end{abstract}
\maketitle

\section{Introduction }

\subsection{Erasure entropies}

 Recently, Verd\`u and Weisman \cite{VW,VW2}, motivated by questions from 
information theory,  introduced and studied   erasure entropies.
In one dimension the description is as follows.
Let $\X=\{X_n\}_{n\in\Z}$ be a stationary process, taking values
in a finite alphabet $A$. The erasure entropy of a collection
of random variables $\{X_1,\ldots,X_n\}$ is defined as
\begin{equation}\label{eq:def}
\H^{-}(X_1,\ldots,X_n) = \sum_{i=1}^n \H(X_i|X_{\setminus i}),
\end{equation}
where $X_{\setminus i}=\{X_j:\ j=1,\ldots,n, \ j\ne i\}$, and 
$\H(X_i|X_{\setminus i})=\H(X_1,\ldots,X_n)-\H(X_{\setminus i})$ is the conditional
entropy.

\begin{defn}[\cite{VW}] \label{erasure_process}\mybf{The erasure entropy rate} of the
process $\X=\{X_n\}_{n\in\Z}$ is defined as
\begin{equation}\label{eq_ee}
\h^{-}(\X) = \lim_{n\to\infty} \frac 1n
\H^{-}(X_1,\ldots,X_n).
\end{equation}
\end{defn}
The erasure entropy measures  the information required
to recover an erased symbol, knowing all the other symbols.

Verd\`u and Weissman established the following properties of erasure entropies.
\begin{thm}[\cite{VW}] The limit in \eqref{eq_ee} exists. Moreover,
\begin{equation}
\h^{-}(\X) = \lim_{n\to\infty} \H(X_0|X_{-n},\ldots,X_{-1},
X_1,\ldots,X_n),
\end{equation}
and
\begin{equation}
\h^{-}(\X)\le \h(\X)= \lim_{n\to\infty} \H(X_0|X_{-1},\ldots,X_{-n }).
\end{equation}
If $\X=\{X_n\}_{n\in\Z}$ is a  $k$-step Markov chain then
\begin{equation}
\h(\X) = \frac{ \h^{-}(\X)+\H(X_1,\ldots,X_k|X_{-1},\ldots,X_{-k})}{k+1}.
\end{equation}
\end{thm}

The interest in objects like the erasure entropy (\ref{eq:def}) and the corresponding
entropy rate (\ref{eq_ee}) arose in information theory in connection to coding problems
for  discrete memoryless erasure channels (DME).

Suppose that $\X=\{X_i\}$ is a discrete stationary process with values in a finite alphabet $\A=\{1,\ldots,N\}$,
and consider the process $\Zpr=\{Z_i\}$ with values in alphabet $\A\cup\{*\}$, defined for all $i$ by
$$
Z_i=\begin{cases} X_i,&\text{ if } E_i=0,\\
*,&\text{ if } E_i=1,
\end{cases}
$$
where $\{E_i\}$ is a sequence of Bernoulli random variables with
$$
\P(E_i=1)=p.
$$
The process $\{E_i\}$ is also assumed to be independent of $\X$. The resulting process $\Zpr$
is the response of the DME with parameter $p$ over $\X$.
The information rate at which the observer of $\Zpr$ has to be supplied in order to reconstruct
the erased symbols almost surely is given by 
$$
\h(\X|\Zpr) =\lim_{n\to\infty} \frac 1n \H(X_1,\ldots,X_n|Z_1,\ldots,Z_n)
$$
and, generically, is very difficult to evaluate. In \cite{VW2}, it was shown that
$$
\h(\X|\Zpr) \ge p\, \h^{-}(\X)\quad \forall p\in[0,1),
$$
and
$$
\h(\X|\Zpr)=p\, \h^{-}(\X)+\mathsf o(p)\quad\text{as }p\to 0.
$$
\section{Bilaterally deterministic processes: 
\newline
source of counterexamples}
A stationary process $\X=\{X_n\}_{n\in\Z}$  is called \mybf{
deterministic} if $X_0$ is measurable with respect to the  $\sigma$-algebra
$\mathcal B(X_k,k<0)$. By stationarity this means that for every $n\in\mathbb N$
$$
\mathcal B(X_k,k\le -n) = \mathcal B(\ldots, X_{-1},X_0,X_{1},\ldots),
$$
and hence $X_0$ can be predicted with zero error probability from its past $(\ldots,X_{-2},X_{-1})$.

Analogously, a stationary process $\X=\{X_n\}_{n\in\Z}$  is called \mybf{bilaterally
deterministic}  if  for every $n\in\mathbb N$,
$$
\mathcal B(X_k,|k|\le n) = \mathcal B(\ldots, X_{-1},X_0,X_{1},\ldots),
$$
which is equivalent to requiring that given the past $\{\ldots,X_{-n-1},X_{-n}\}$
and  the  future $\{X_{n},X_{n+1},\ldots\}$, the value $X_0$ can be reconstructed
with zero error probability. Therefore, any bilaterally deterministic process has
erasure entropy rate $0$.

Gurevi{\v{c}} \cite{G} constructed a first example of a non-deterministic, but 
bilaterally deterministic, stochastic process;
for a strongly mixing example see \cite{BDS}.
Any deterministic process is clearly bilaterally deterministic.
The converse is  false  in quite a strong sense, as the following somewhat counter-intuitive result obtained by Ornstein and Weiss \cite{OW} shows.

\begin{thm}\label{OWres}
Given any ergodic finite-state stationary stochastic process, there
is an isomorphic ergodic finite-state stationary stochastic process which is
bilaterally deterministic.
\end{thm}

Isomorphism here is understood in measure-theoretical sense. More specifically,
in \cite{OW} a finite {\it generating} partition $\beta=\{B_1,\ldots,B_m\}$ of $A^{\Z}$ is built such that
the corresponding factor process $\mathcal Y=\{Y_n\}$ given by
$$
Y_n=j\quad \iff \quad \sigma^{n}(\{X_k\}_{k\in\Z})\in B_j\quad \text{for all } n,
$$
where $\sigma:A^\Z\to A^\Z$ is the left shift,
 is bilaterally deterministic. The fact that the partition $\beta$ is generating
 means that given $\{Y_n\}$ one can reconstruct $\{X_n\}$ almost surely.

Theorem \ref{OWres} implies that the erasure entropy rate $\h^{-}(\X)$ is not
a measure-theoretic invariant of the process $\X=\{X_n\}$, while the entropy
rate $\h(\X)$ is  well known to be such an invariant.

\section{Erasure entropies of Gibbs states}
 A natural relation between
erasure entropies and  thermodynamic (Gibbs)  states was observed in the 
original paper \cite{VW}. In particular, the application to Gibbs sampling 
of Markov random fields, a well-known problem in stochastic image analysis, 
served as a motivation for considering this question also on 
higher-dimensional lattices, see in particular their section {\bf 4F}, in 
which they provide an approximate expression for the erasure entropy of a 
two-dimensional nearest-neighbour Ising model.
Let us elaborate on this issue a bit further.

\subsection{Gibbs measures}
Gibbs measures for a given interaction $\Phi$, according to the definition introduced by 
Dobrushin, Lanford and Ruelle, are defined as measures whose conditional 
probabilities of finite-volume configurations $X_{\Lambda}$, given 
external configurations $X_{\Lambda^c}$, are of the the Gibbsian form
$$
\P(X_\Lambda|X_{\Lambda^c}) = \frac 1{Z_{X_{\Lambda^c}}} \exp\Bigl( -H_\Lambda^\Phi(X)\Bigr)=:
\gamma_\Lambda^{\Phi}(X_\Lambda|X_{\Lambda^c})
$$
where the Hamiltonian
$H_{\Lambda}^\Phi$ contains all interaction terms having a non-empty 
intersection with $\Lambda$:
$$
H_\Lambda^\Phi(X)=\sum_{A\cap\Lambda\ne\varnothing} \Phi_A(X),
$$
where the normalization constant $Z_{X_{\Lambda^C}}$ is the appropriate
partition function.
\newline
This should hold for all volumes $\Lambda \subset \Zd$, and all choices of the 
configurations $X_{\Lambda}$ and $X_{\Lambda^c}$. See e.g. \cite{Geo,EFS}.

\subsection{Erasure entropy rates of random fields}
For a finite volume $\Lambda\subseteq\Z^d$ and any
$\mathbf k\in\Z^d$, we describe the shifted volume $\mathbf k+\Lambda=\{\mathbf k
+\mathbf i|\mathbf i\in\Lambda\}$.

\begin{defn}\label{def_ee_Zd}
Let $\{X_n\}_{n\in\Zd}$ be a translation-invariant
random field taking values in
a finite alphabet $A$, and let $\P$ be the corresponding
probability measure on $A^\Zd$.
The erasure entropy rate of $\P$ for the finite volume $\Lambda\subset\Z^d$
is defined as
\begin{equation}\label{defZd}
\h_\Lambda^{-}(\P) = \lim_{V\nearrow\Z^d} \frac 1{|V|}
\sum_{\mathbf k\in V} \H\bigl( X_{(\mathbf k+\Lambda)\cap V}|X_{V\setminus(\mathbf
k+\Lambda)}\bigr)
\end{equation}
\end{defn}
The interpretation of the erasure entropy for a finite volume remains unchanged: $\h_\Lambda^{-}(\P)$ is the number of bits required to
reconstruct the configuration on $\Lambda$ given the configuration outside of $\Lambda$. In fact, it is also convenient to consider
the  erasure entropy  normalized by the volume
$\frac 1{|\Lambda|}{\h}^{-}_\Lambda(\P)$.

\begin{thm}\label{formula} The limit (\ref{defZd}) exists, and
\begin{equation}\label{lim1}
\h_\Lambda^{-}(\P) =\lim_{V\nearrow\Zd} \H(X_\Lambda|X_{V\setminus\Lambda})
=-\int \log \P(X_\Lambda|X_{\Lambda^c}) \P(dX).
\end{equation}
\end{thm}

For Gibbs measures one can say more.

\begin{thm}\label{gibbsformula}
If $\P$ is a translation invariant Gibbs measure on $A^\Zd$
for interaction $\Phi$, then
the erasure entropy rate can be expressed in terms of the
Gibbs specification $\gamma_\Lambda^\Phi$ as follows:
\begin{equation}\label{lim2}
{\h}^{-}_\Lambda(\P) =-\int \log \gamma_{\Lambda}^{\Phi}(X_\Lambda|X_{\Lambda^c})\, \P(dX).
\end{equation}
Hence, ${\h}^{-}_\Lambda(\cdot)$ is an affine functional
on the set of Gibbs measures consistent with specification $\gamma^\Phi$.
Moreover,
\begin{equation}\label{lim3}
\lim_{\Lambda\nearrow\Zd} \frac {1}{|\Lambda|}
{\h}^{-}_\Lambda(\P)=
\h(\P).
\end{equation}
\end{thm}
\begin{proof}
Proofs of such results 
are rather standard and rely on application of (strong) subadditivity, 
martingale convergence and ergodic theorems, \cite{VW,vE,T}.
Expression (\ref{lim2}) is an immediate consequence of the
Gibbs property of $\P$. To show (\ref{lim3}) we note that
$$\aligned
-\frac 1{|\Lambda|}\int &
\log \gamma_{\Lambda}^{\Phi}(X_\Lambda|X_{\Lambda^c})\, \P(dX)\\
&=\int \frac 1{|\Lambda|}H_\Lambda^\Phi(X)\P(dX)+\int
 \frac 1{|\Lambda|}\log Z(X_{\Lambda^c}) \P(dX)\to \h(\P)
\endaligned
$$
by \cite{G}*{Corollary 15.35}.
\end{proof}

Let us conclude this section with a number of remarks on the 
relation between erasure entropies and  the theory of Gibbs states of 
Equilibrium Statistical
Mechanics.

\medskip
(A) By (\ref{lim1}), the erasure entropy rate can  also be seen as the difference in entropies of
the system with and without a finite number of spins inside an interval or a volume.
As such, erasure entropy is an example of an \emph{inner entropy} 
(in some literature called local or conditional entropy ) introduced earlier
in Statistical Mechanics, see \cite{FS,S} or \cite{T, T2}.

Such conditional entropies have been used before to establish ``local'' or
``inner'' variational principles \cite{BR, S, S2, T, T2}. 
From (\ref{lim3}) it 
is known that for Gibbs measures  the limit conditional entropy density and the
standard entropy density coincide \cite{vE,T2}. This property is obviously not true in general,  due to the Gurevi{\v{c}} and Ornstein-Weiss
examples, see e.g. \cite{T, T2}.

\medskip
(B) One extra remark to be made  is that Gibbs
measures can be characterised, not only by a local variational principle for 
general volumes, but even by a single-site variational principle,
using only the (Verd\`u-Weissman) single-site erasure entropy.
This was proven in a short paper by J. Fernando Perez and R.H. Schonmann 
\cite{FP-S}. More specifically, suppose $\Phi=\{\Phi_\Lambda\}$ is an interaction and $\P$ is an arbitrary
probability measure (state). The \emph{free-energy content of a region $\Lambda$ in a state $\P$}
is defined by
$$
F_\Lambda(\P)=\int H_\Lambda d\P+ \h_\Lambda^{-}(\P).
$$
Here we do not assume $\P$ to be translation invariant, and hence the erasure entropy $\h^{-}_\Lambda(\P)$
is understood as the limit 
$$
\h^{-}_\Lambda(\P)=\lim_{V\nearrow \Z^d} \H(X_\Lambda|X_{V\setminus\Lambda}).
$$
Following Sewell \cite{S2}, a state $\P$ is called \mybf{locally thermodynamically stable} if
for every finite volume $\Lambda$ one has
$$
F_\Lambda(\tilde \P) \ge F_\Lambda( \P)
$$
for all measures $\tilde\P$ such that $\tilde\P_{\Lambda^c}=\P_{\Lambda^c}$.
In other words, the locally thermodynamically stable state minimizes the free-energy content
with respect to variations of the state inside $\Lambda$, for all $\Lambda$.

It turns out that for  a fairly large class of finite-norm interactions, from the fact that $\P$ minimizes
$F_\Lambda$, one can deduce that $\P$ satisfies the  Dobrushin-Lanford-Ruelle (DLR)
conditions on $\Lambda$, see \cite[Theorem 2]{FP-S}. 
Moreover, it also known that if a translation-invariant state $\P$ satisfies DLR condition on $\Lambda=\{0\}$, then $\P$ is a Gibbs state. (Or, more generally, if the DLR conditions hold for all single sites they hold in general).

Hence,  we  can formulate a variational principle for Gibbs states in terms of the erasure entropy $\h_{0}^-(\P)$.

Although \cite{FP-S} formulate their result for finite-range interactions, it is easy to see that their arguments generalise beyond that, due to the  analysis in \cite{S,BR} of the LTS condition and its equivalence to other equilibrium conditions.

\medskip
(C) It follows from the bilaterally deterministic examples discussed in 
Section 2 that neither the erasure entropy, nor the Gibbs property are 
measure-theoretical invariants, c.f. \cite{dHS}.


\section{ Computing erasure entropies 
for Ising models in two dimensions}
In \cite{VW2} a question arose about computation of a 1-site
(i.e., $\Lambda=\{0\}$) erasure entropy rate for the standard nearest-neighbour
Ising model on the two-dimensional integer lattice. 
In fact,  in contrast to what was suggested there, 
this erasure entropy rate can be computed explicitly. 

Because of the Markov property of the Gibbs measures of the 2-dimensional
Ising model
(immediate from the nearest-neighbour character of the
interaction), it is sufficient to consider
nearest-neighbour environments of the spin at the center. Let us denote
the 5 spins we will need to consider
$\sigma_C, \sigma_E, \sigma_N,\sigma_W,\sigma_S$
(for Center, East, North, West, and South).

\begin{figure}[ht]
\begin{tabular}{|c|c|c|}
\hline
 & $\sigma_{N}$ &\\ \hline
$\sigma_{W}$ & $\sigma_C$ & $\sigma_E$\\\hline
& $\sigma_{S}$ &\\ \hline
\end{tabular}
\caption{Ising spins}
\end{figure}

We first remark that the probabilities of the 16 configurations of
the E., N., W., and S. spins, in the spin-flip-invariant infinite-volume Gibbs measure 
$\P= \frac{1}{2}
(\P^{+} + \P^{-})$,  fall into
4 classes.

\begin{figure}[ht]
\begin{tabular}{|c|c|c|}
\hline
 & $+$ &\\ \hline
$+$ & $\cdot$ & $+$\\
\hline
& $+$ &\\ \hline
\end{tabular}
\begin{tabular}{|c|c|c|}
\hline
 & $+$ &\\ \hline
$+$ & $\cdot$ & $-$\\
\hline
& $+$ &\\ \hline
\end{tabular}
\begin{tabular}{|c|c|c|}
\hline
 & $+$ &\\ \hline
$-$ & $\cdot$ & $-$\\
\hline
& $+$ &\\ \hline
\end{tabular}
\begin{tabular}{|c|c|c|}
\hline
 & $+$ &\\ \hline
$+$ & $\cdot$ & $-$\\
\hline
& $-$ &\\ \hline
\end{tabular}
\caption{Four boundary types (upto rotations and spin-flips)}
\end{figure}

Denote by $P_i$, $i=1,\ldots,4$ the probability of the corresponding boundary configuration.
The appropriate  combination of these $4$ probabilities adds up to $1$:
\begin{equation}\label{E1}
2P_1+8P_2+2P_3+4P_4=1.
\end{equation}
By (\ref{lim2}), the erasure entropy $\h_{\{0\}}^-(\P)$ is given by:
$$\aligned
\h_{\{0\}}^-(\P)=& 2P_1 \cdot h\Bigr( \frac{ e^{4J}}{e^{4J}+e^{-4J}}\Bigr)+8P_2\cdot h\Bigl(\frac{ e^{2J}}{e^{2J}+e^{-2J}}\Bigr)\\
&\phantom{1111111}+(2P_3+4P_4)\cdot h\Bigl(\frac 12\Bigr),\\
\endaligned
$$
where $h(p)=-p\log p-(1-p)\log(1-p)$ for $p\in [0,1]$, and $J$ 
is the nearest-neighbour coupling constant.
Note that in \cite{VW2} the third and the fourth boundary types are not distinguished (since
they contribute equally to the entropy). However, to actually compute the erasure entropy $\h_{\{0\}}^-(\P)$ we must distinguish these types.

By Pfaffian methods (albeit in somewhat implicit form, see \cite{KBT}), one can compute   three relevant correlation functions, namely the 4-point
function and the
two pair correlation functions one can build with these four spins,  that is
$\E(\sigma_E \sigma_W),
\E(\sigma_E \sigma_N)$,
and $\E (\sigma_E \sigma_N \sigma_W \sigma_S)$.
Finally, we derive remaining equations for probabilities $P_i$, by noting that
\begin{eqnarray}
\E (\sigma_E \sigma_N \sigma_W \sigma_S) &=& 2P_1 -8P_2 +2P_3 +4P_4,     \label{E2} \\
\E(\sigma_E \sigma_W) &=&  2P_1 +2 P_3 -4P_4,  \label{E3}  \\
 \E(\sigma_E \sigma_N) &=& 2P_1  -2P_3. \label{E4}
\end{eqnarray}
As an immediate consequence, one gets
$$\aligned
P_1&=\frac{1+\E (\sigma_E \sigma_N \sigma_W \sigma_S) +2\E (\sigma_E \sigma_W)+
4 \E (\sigma_E \sigma_N)}{16},\\
P_2&=\frac{1-\E (\sigma_E \sigma_N \sigma_W \sigma_S) }{16},\\
\endaligned
$$
which  is sufficient to express the erasure entropy in terms of correlation functions.

In conclusion, we note that for obvious symmetry reasons
$$
\h_{\{0\}}^-(\P^-)=\h_{\{0\}}^-(\P^+),
$$
and,  by the affine property, the erasure entropies of both phases  coincide with $\h_{\{0\}}^-(\P)$ obtained above. For the same reasons, the erasure entropy of any Gibbs measure of the Ising model  coincides with $\h_{\{0\}}^-(\P)$.

\bigskip

Definitions of erasure entropies extend readily to more general lattices. 
For the correspoding nearest-neighbour Ising models the Pfaffian method of computing correlations 
will apply to some other two-dimensional
planar lattices, see again \cite{BKT, BTKM, KBT} and specifically \cite{BMThoney}.
In case of the hexagonal lattice, we propose another approach, 
which is less general, in the sense that everything we need 
can be derived from just knowing the free energy density, but already  
leads to an explicit analytic expression of the erasure entropy.

Consider the two-dimensional hexagonal (honeycomb) lattice. One now needs to consider only three neighbours -- denoted $\sigma_I$, $\sigma_{II}$,  and 
$\sigma_{III}$, of the central spin $\sigma_C$. There are obviously   only 2 boundary types.

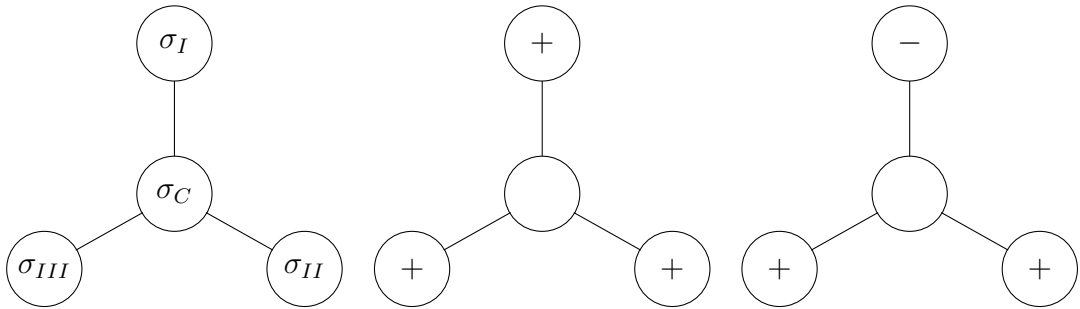
\begin{figure}[h]
 \begin{tikzpicture} 
\draw (0,.5) -- (0,1.5);  \draw (-0.42,-.25) -- (-1.31,-.75);  \draw (0.42,-.25) -- (1.31,-.75); 
\draw (0,0) circle (0.5cm); \draw(0,0) node {$\sigma_C$};
\draw (0,2) circle (0.5cm); \draw(0,2) node {$\sigma_{I}$};
\draw(1.73,-1) circle (0.5cm); \draw(1.73,-1) node {$\sigma_{II}$};
\draw(-1.73,-1) circle (0.5cm); \draw(-1.73,-1) node {$\sigma_{III}$};
\end{tikzpicture}\quad
\begin{tikzpicture} 
\draw (0,.5) -- (0,1.5);  \draw (-0.42,-.25) -- (-1.31,-.75);  \draw (0.42,-.25) -- (1.31,-.75); 
\draw (0,0) circle (0.5cm);  
\draw (0,2) circle (0.5cm); \draw(0,2) node {$+$};
\draw(1.73,-1) circle (0.5cm); \draw(1.73,-1) node {$+$};
\draw(-1.73,-1) circle (0.5cm); \draw(-1.73,-1) node {$+$};
\end{tikzpicture}\quad
\begin{tikzpicture} 
\draw (0,.5) -- (0,1.5);  \draw (-0.42,-.25) -- (-1.31,-.75);  \draw (0.42,-.25) -- (1.31,-.75); 
\draw (0,0) circle (0.5cm);  
\draw (0,2) circle (0.5cm); \draw(0,2) node {$-$};
\draw(1.73,-1) circle (0.5cm); \draw(1.73,-1) node {$+$};
\draw(-1.73,-1) circle (0.5cm); \draw(-1.73,-1) node {$+$};
\end{tikzpicture}
\caption{Ising spins on hexagonal lattice and two  boundary types.}
\end{figure}

Consider again the zero-field situation and the symmetric Gibbs measure $\P$.
Let $P_i$, $i=1,2$, denote the probability of the corresponding boundary type.
Thus we have
\begin{equation}\label{hex1}
2P_1 + 6 P_2 =1.
\end{equation}
For each quadruple $(\sigma_{I},\sigma_{II},\sigma_{III},\sigma_{C})\in\{-1,1\}^4$,
one has
$$\aligned
\P(\sigma_{I},\sigma_{II},\sigma_{III},\sigma_{C})&=\P(\sigma_{C}|\sigma_{I},\sigma_{II},\sigma_{III})\P(\sigma_{I},\sigma_{II},\sigma_{III})\\
&=\frac {e^{J\sigma_C(\sigma_{I}+\sigma_{II}+\sigma_{III})}}{2\cosh\left( J(\sigma_{I}+\sigma_{II}+\sigma_{III})\right)}P_i,
\endaligned
$$
where the boundary type $i$ is determined by $(\sigma_{I},\sigma_{II},\sigma_{III})$.

Consider now the  correlation function $\E(\sigma_I\sigma_{C})$,
$$
\E(\sigma_I\sigma_{C})=\sum_{(\sigma_{I},\sigma_{II},\sigma_{III},\sigma_{C})\in\{-1,1\}^4}
\sigma_{I}\sigma_{C}\P(\sigma_{I},\sigma_{II},\sigma_{III},\sigma_{C}),
$$ 
and after a straightforward manipulation we conclude  that
$$
\E(\sigma_I\sigma_{C})= P_1\tanh(3J)+P_2\tanh(J).
$$ 
Hence,
\begin{equation}\label{HC0}
P_1=\frac {6\E(\sigma_I\sigma_{C})-\tanh(J)}{6\tanh(3J)-2\tanh(J)},
\quad
P_2=\frac {\tanh(3J)-2\E(\sigma_I\sigma_{C})}{6\tanh(3J)-2\tanh(J)},
\end{equation}
and the erasure entropy $\h^{-}_0(\P)$ 
\begin{equation}\label{HC1}\aligned
\h^{-}_0(\P)&= 2P_1\cdot h\Bigl(\frac {e^{3J}}{e^{3J}+e^{-3J}}\Bigr)+6P_2\cdot h\Bigl(\frac {e^{J}}{e^{J}+e^{-J}}\Bigr)
\endaligned
 \end{equation}
can be expressed in terms of the correlation function $\E(\sigma_{I}\sigma_{C})$.

In order to obtain the expression for $\E(\sigma_{I}\sigma_{C})$ we recall that the zero-field
Ising model on the hexagonal lattice is an exactly solvable model, and the
pressure function at inverse temperature $\beta$ is given by 
\begin{equation}\label{HC2}\aligned
p(\beta) =\frac 34\log 2+ \frac 1{16\pi^2}\int_{0}^{2\pi}\int_{0}^{2\pi}
\log \Bigl( \cosh^3&(2\beta J)+1-\sinh^2(2\beta J)\times\\
&\Bigl[\cos(\theta_1-\theta_2)+\cos(\theta_1)+\cos(\theta_2)\Bigr]\Bigr)d\theta_1
d\theta_2,
\endaligned
 \end{equation}
see e.g. \cite[eq. (A.6). p.322]{Sy}. The pressure function $p(\beta)$ 
is an everywhere  differentiable
function of $\beta$; the phase transition in the Ising model on hexagonal 
lattice is of  second order.

Finally, as  is well known \cite{Geo,Rue}, Gibbs states are tangent functionals
to the pressure, and hence, if the interaction $\Phi=\{\Phi_\Lambda\}$  
is such that the pressure function $p(\beta)=P(\beta\Phi)$ is
differentiable as a function of $\beta$ at $\beta=\beta_0$, then
\begin{equation}\label{HC3}
p'(\beta_0) = \int \Bigl( -\sum_{\Lambda\ni 0} \frac 1{|\Lambda|} \Phi_\Lambda(\sigma) \Bigr) \P(d\sigma),
 \end{equation}
where $\P$ is a Gibbs state for the Hamiltonian $\beta H=\beta\sum_{\Lambda}\Phi_\Lambda$.
 In case of the Ising model on the hexagonal lattice, we thus conclude that 
\begin{equation}\label{HC4}
 p'(1)= -\frac {3J}{2} \, \E(\sigma_I\sigma_{C}) .
 \end{equation}
 Combining (\ref{HC4}), (\ref{HC2}), (\ref{HC0}), and (\ref{HC1}), we arrive to an analytic
 expression for $\h_{0}^{-}(\P)$. Similarly to the case of rectangular lattice,
 the erasure entropy is independent of the phase.

{\bf Final remarks} A possible generalisation rests on the fact that the coupling constants in the three different directions are allowed to be different, without affecting the solubility of the model. In that case the corresponding correlation functions need to be separately computed.

Also, the next nearest-neighbour correlation functions  
$\E(\sigma_I \sigma_{II})$ etc, can be computed as the 
{\em nearest neighbour} correlation functions 
of an Ising model on a triangular lattice, related to the hexagonal model 
by  a star-triangle transformation, see again \cite{Sy}.  This offers a
slightly different, though essentially equivalent, route to computing 
exact expressions for the Verd\`u-Weissman erasure entropy  

{\bf Acknowledgements} We thank J.H.H. Perk for helpful advice
about the literature and M. Khatun for a helpful correspondence.

 

\begin{bibdiv}
 \begin{biblist}
\bib{BTKM}{article}{
author={J.H. Barry}, 
author={T. Tanaka}, 
author={M. Khatun}, 
author={ C.H. M\'unera}, 
title= {Exact solutions for  Ising-model odd-number correlations on planar lattices},
journal={Phys. Rev.B},
volume={ 44}, 
date={1991},
pages={2595--2608},
}
\bib{BMThoney}{article}{
author={J.H. Barry}, 
author={C. H. M\'unera},
author={T. Tanaka}, 
title= {Exact solutions for Ising model odd-number correlations on the honeycomb and triangular lattices},
journal={Physica A: Statistical and Theoretical Physics},
volume={113},
number={3}, 
date={1982},
pages={367--387},
}
 

\bib{BKT}{article}{
   author={Barry, J. H.},
   author={Khatun, M.},
   author={Tanaka, T.},
   title={Exact solutions for Ising-model even-number correlations on planar
   lattices},
   journal={Phys. Rev. B (3)},
   volume={37},
   date={1988},
   number={10},
   pages={5193--5204},
   issn={0163-1829},
   review={\MR{933704 (89b:82102)}},
}


\bib{BDS}{article}{
   author={Burton, R. M.},
   author={Denker, M.},
   author={Smorodinsky, M.},
   title={Finite state bilaterally deterministic strongly mixing processes},
   journal={Israel J. Math.},
   volume={95},
   date={1996},
   pages={115--133},
   issn={0021-2172},
   review={\MR{1418290 (98b:60069)}},
}
\bib{BR}{book}{
    author={Bratteli, O.},
    author={Robinson, D.W.},
    title ={Operator Algebras and Quantum Statistical Mechanics, vol.II, second edition},
    pages={274--284}, 
publisher={Springer},
 }
 
\bib{G}{article}{
   author={Gurevi{\v{c}}, B. M.},
   title={One- and two-sided regularity of stationary random processes},
   language={Russian},
   journal={Dokl. Akad. Nauk SSSR},
   volume={210},
   date={1973},
   pages={763--766},
   issn={0002-3264},
   review={\MR{0322939 (48 \#1299)}},
}
\bib{vE}{article}{
   author={van Enter, A. C. D.},
   title={ On a Question of Bratteli and Robinson},
   journal={ Lett.Math.Phys.},
   volume= {6},
   date={ 1982},
   pages={289-291},
}
\bib{EFS}{article}{
   author={van Enter, A. C. D.},
   author={Fern\'andez, R.},
   author={Sokal, A.D.},
title={Regularity Properties and Pathologies of Position-Space renormalization-Group transformations: Scope and Limitations of Gibbsian Theory},
journal={J. Stat.Phys.},
volume={72},
date={1993},
pages={879-1169},
}
\bib{FP-S}{article}{
author={J. Fernando Perez},
author={R.H. Schonmann},
title={On the global character of some restricted equilibrium conditions - a remark on metastability},
journal={J. Stat.Phys.},
volume={28},
date={1982},
pages={479-485},
}


\bib{FS}{article}{
author={F\"ollmer, H.},
author={Snell, J.L}, 
title={An ``Inner'' Variational Principle for Markov Fields.},
journal={Prob.Th. Rel. Fields},
volume={ 39}, 
pages={187--195},
date={ 1977},
}
\bib{Geo}{book}{
   author={Georgii, H.-O.},
   title={Gibbs measures and phase transitions},
   series={de Gruyter Studies in Mathematics},
   volume={9},
   publisher={Walter de Gruyter \& Co.},
   place={Berlin},
   date={1988},
   pages={xiv+525},
   isbn={0-89925-462-4},
   review={\MR{956646 (89k:82010)}},
}

\bib{KBT}{article}{
 author={Khatun, M.},
   author={Barry, J. H.},
   author={Tanaka, T.},
   title={Exact solutions for even-numbered correlations of the square Ising model},
   journal={Phys. Rev. B },
   volume={42},
   date={1990},
   pages={4398--4406},
}

\bib{dHS}{article}{
author={F. den Hollander},
author={J.E. Steif},
title={On the equivalence of certain ergodic properties for Gibbs states},
journal={Ergodic Theory and Dynamical Systems},
Volume={20},
date={2000},
pages={231-239},
}

\bib{OW}{article}{
   author={Ornstein, D. S.},
   author={Weiss, B.},
   title={Every transformation is bilaterally deterministic},
   note={Conference on Ergodic Theory and Topological Dynamics (Kibbutz
   Lavi, 1974)},
   journal={Israel J. Math.},
   volume={21},
   date={1975},
   number={2-3},
   pages={154--158},
   issn={0021-2172},
   review={\MR{0382600 (52 \#3482)}},
}


\bib{Rue}{book}{
   author={Ruelle, D.},
   title={Thermodynamic formalism},
   series={Encyclopedia of Mathematics and its Applications},
   volume={5},
   note={The mathematical structures of classical equilibrium statistical
   mechanics;
   With a foreword by Giovanni Gallavotti and Gian-Carlo Rota},
   publisher={Addison-Wesley Publishing Co., Reading, Mass.},
   date={1978},
   pages={xix+183},
   isbn={0-201-13504-3},
   review={\MR{511655 (80g:82017)}},
}
\bib{S}{article}{
author={Sewell, G.~L.},
title={Statistical mechanical theory of metastable states},
journal={Lett. Nuovo Cim.},
volume={10},
pages={ 430--434}, 
date={1974},
}

\bib{S2}{article}{
author={Sewell, G.~L.},
title={Stability, equilibrium and  metastability in statistical mechanics},
journal={Phys.Rep.},
volume={57},
pages={307-342},
date={1980},
}

\bib{Sy}{article}{
author={Syozi, I.},
title={Transformations of Ising models},
book={
series={Phase Transitions and Critical Phenomena},
volume={1},
editor={C.Domb},
editor={M.~S.~Green},
date={1972},
},
}
\bib{T}{book}{
   author={Tempelman, A.A.},
   title={Ergodic theorems for group actions},
   series={Mathematics and its Applications},
   volume={78},
   note={Informational and thermodynamical aspects;
   Translated and revised from the 1986 Russian original},
   publisher={Kluwer Academic Publishers Group},
   place={Dordrecht},
   date={1992},
   pages={xviii+399},
   isbn={0-7923-1717-3},
   review={\MR{1172319 (94f:22007)}},
}


\bib{T2}{article}{
   author={Tempelman, A. A.},
   title={Specific characteristics and variational principle for homogeneous
   random fields},
   journal={Z. Wahrsch. Verw. Gebiete},
   volume={65},
   date={1984},
   number={3},
   pages={341--365},
   issn={0044-3719},
   review={\MR{731226 (85m:60158)}},
}
		

	


\bib{VW}{article}{
title={Erasure entropies},
author={Verd\`u, S.},
author={Weissman, T.},
journal={Proceedings ISIT 2006},
pages={98--102},
}
\bib{VW2}{article}{
title={The Information Lost in Erasures},
author={Verd\`u, S.},
author={Weissman, T.},
journal={IEEE Trans. Inf. Theory},
volume={54},
number={11},
year={2008},
pages={5030-5058},
}

 \end{biblist}
 \end{bibdiv}
\end{document}